\documentclass[preprint,12pt]{elsarticle}
\usepackage{amsmath}
\usepackage{amssymb}
\usepackage{multirow}
\usepackage{enumerate}
\usepackage{cases}
 \usepackage{amsthm}
 \usepackage{framed}
\usepackage{tabularx}
\usepackage{mathrsfs}

\newtheorem{Theorem}{Theorem}

\usepackage{algorithm}%
\usepackage{tabularx}
\usepackage{listings}%
\usepackage{caption}
\usepackage{mathtools} 
\usepackage{algpseudocode}
\usepackage{graphicx}
\usepackage{slashed}

\makeatletter
\newcommand{\multiline}[1]{%
  \begin{tabularx}{\dimexpr\linewidth-\ALG@thistlm}[t]{@{}X@{}}
    #1
  \end{tabularx}
}
\makeatother

\begin{document}

\begin{frontmatter}
\title{Quantum Truncated Differential and Boomerang Attack}


\author{Huiqin Xie$^{1}$\corref{1}}
\author{Li Yang$^{2}$}
\cortext[1]{Corresponding author email: xiehuiqindky@163.com}
\address{1. Department of Cryptography Science and Technology, Beijing Electronic Science and Technology Institute, Beijing 100070, China\\
2. Institute  of  Information  Engineering, Chinese Academy of Sciences, Beijing 100085, China}

\begin{abstract}
Facing the worldwide steady progress in building quantum computers, it is crucial for cryptographic community to design quantum-safe cryptographic primitives. To achieve this, we need to investigate the capability of cryptographic analysis tools when used by the adversaries with quantum computers. In this article, we concentrate on truncated differential and boomerang cryptanalysis. We first present a quantum algorithm which is designed for finding truncated differentials of symmetric ciphers. We prove that, with a overwhelming probability, the truncated differentials output by our algorithm must have high differential probability for the vast majority of keys in key space. Afterwards, based on this algorithm, we design a quantum algorithm which can be used to find boomerang distinguishers. The quantum circuits of both quantum algorithms contain only polynomial quantum gates. Compared to classical tools for searching truncated differentials or boomerang distinguishers, our algorithms fully utilize the strengths of quantum computing, and can maintain the polynomial complexity while fully considering the impact of S-boxes and key scheduling.

\end{abstract}


\end{frontmatter}


\section{Introduction}

Recently, research on quantum computers has continuously made new progress worldwide. Many scientists, companies and research institutions are committed to utilizing various quantum systems to achieve quantum computers. It can be foreseen that the successful development of quantum computers will have profound impacts on many fields. Cryptography is one of them. 

The great superiority of quantum computers for information processing stems from the novel properties of quantum information which differ from classical information. Quantum computers have the natural feature of parallel computing. When a $n$-qubit quantum computer processes data, its one operator is in fact operating on $2^n$ data states simultaneously. This parallelism may make some problems uncomputable in electronic computers become computable in quantum computers, such as factoring large integers, which is a difficult problem that many public key algorithms are built upon, but may be solved on quantum computers by running Shor's algorithm \cite{ref-proceeding1} . 

The threat posed by quantum computing to symmetric algorithms has also received attention. The most typical example is Grover’s algorithm \cite{ref-proceeding2}, which requires only $O(\sqrt{M})$ complexity to search a unordered database with $M$ elements, while $O(M)$ complexity is required in classical computing. Another important algorithm used to attack symmetric schemes is Simon’s algorithm \cite{ref-journal1}. It was first used for attacking Feistel ciphers \cite{ref-proceeding3,ref-journal2,ref-proceeding5} and EM schemes \cite{ref-proceeding3,ref-proceeding5}, then it is also combined with Grover’s algorithm for extracting the key of ciphers with FX, Feistel and generalized Feistel structures \cite{ref-proceeding6,ref-journal3,ref-journal4}. For SPN ciphers, Jaques et al. investigated the cryptanalysis of AES algorithm via Grover’s algorithm \cite{ref-proceeding10}. Besides the above mentioned quantum algorithms, Bernstein-Vazirani algorithm \cite{ref-journal10} was also utilized in cryptanalysis recently \cite{ref-proceeding16,ref-journal11,ref-proceeding17,ref-proceeding18}.

In addition to specific attack strategies, cryptanalytic tools are also crucial for evaluating the security of cryptosystems. In this field, quantum algorithms were first used for differential attack \cite{ref-journal7,ref-proceeding10,ref-journal11}, and then used for linear attack \cite{ref-proceeding10,ref-proceeding17,ref-journal13}. Afterwards, the quantum collision attacks to Hash functions are studied \cite{ref-proceeding11,ref-proceeding12}. There were also quantum attacks under quantum related-key model proposed \cite{ref-journal5,ref-journal6,ref-journal12}. These attacks showcased the superiority of quantum algorithms when applied to traditional cryptanalytic tools.   

\textbf{Contributions.} In this paper, we explore the applications of Bernstein-Vazirani algorithm in two variants of differential attack: truncated differential and boomerang attacks. We first design a quantum algorithm which is for searching truncated differentials that have high probability for a large proportion of keys in key space. Afterwards, based on this algorithm, we construct another quantum algorithm which is for finding boomerang distinguishers. We demonstrate the correctness of both quantum algorithms by rigorous proofs. Both quantum algorithms merely request polynomial quantum gates and qubits, and have the advantages:

\begin{enumerate}[]
\item Quantum adversaries are able to perform the proposed attacks in $Q_1$ model. Namely, there is no need for quantum queries. Compared to many proposed quantum attack algorithms \cite{ref-proceeding3,ref-proceeding4,ref-journal2,ref-proceeding5,ref-journal3,ref-journal4} that require quantum queries, our algorithms are easier to implement. 
\vskip 0.2cm

\item Classical cryptanalytic tools for finding truncated differentials with high probability or boomerang distinguishers usually cannot concern all details of the involved S-boxes when they are not small-scale. The classical tools can only search for truncated differentials or boomerang distinguishers of extremely few rounds when the S-boxes have 8-bit scale, which is very common in block ciphers. By comparison, our quantum algorithms fully utilize the superiority of quantum computing to improve this issue. They entirely characterized the S-boxes through the accurate implementation of the unitary operator of the block ciphers. 
\vskip 0.2cm

\item Classical truncated differential attacks do not involve the key scheduling under single-key attack model, but our algorithms incorporate the key scheduling into the quantum circuits and thus fully reflect its impact to the differential propagation.
\end{enumerate}

\section{Preliminaries}
The main notations and their definitions are presented in Table 1.
\begin{table}[H] 
    \caption{Notations.\label{Tab1}}
\begin{tabular}{ll}
\hline
\textbf{Notation}	& \textbf{Definition}\\
\hline
$\mathcal{C}_{u,v}$		&  the set containing all Boolean functions mapping $u$ bits to $v$ bits \\
$\Phi$		& the empty set \\
$S_f(\cdot)$ & the Walsh transform of $f$\\
$Enc_k^t$&  a $t$-round block cipher\\
\vspace{0.1cm}
$Enc_k^t[j]$& the $j$-th component function of $Enc_k^t$\\
\vspace{0.1cm}
$S/N$& the ratio of signal to noise\\
$(\Delta x,\Delta y)$& a differential\\
\vspace{0.1cm}
$(\overline{\Delta} x,\overline{\Delta} y)$& a truncated differential\\
\hline
\end{tabular}
\end{table}
\subsection{Differential}
Throughout this paper, $Enc_k:\mathbb{F}_2^n\rightarrow\mathbb{F}_2^n$ denotes a block cipher, where $k\in\mathbb{F}_2^m$ is the master key. $Enc_k^{r}=Enc_{k_r}\circ Enc_{k_{r-1}}\circ\cdots\circ Enc_{k_1}$ denotes the $r$-round iteration of $Enc_k$. Here $k_1,\cdots,k_r$ denote the round keys generated from $k$ according to the key scheduling.   

Suppose $x$ and $x'$ are two plaintexts, and $Enc_k^r(x)=y$, $Enc_k^r(x')=y'$. We call $\Delta y=y'\oplus y$ an output difference and $\Delta x=x'\oplus x$ an input difference. $(\Delta x,\Delta y)$ is defined as a differential of $Enc_k^r$. The probability of differential $(\Delta x,\Delta y)$ is defined as
\begin{align*}
\Pr_{ x\leftarrow\mathbb{F}_2^n}[\Delta x\overset{Enc_k^r}{\rightarrow}\Delta y]&=\Pr_{x\leftarrow\mathbb{F}_2^n}[Enc_k^r(x\oplus \Delta x)\oplus Enc_k^{r}(x)=\Delta y]\\
&=\frac{1}{2^{n}}|\{x\in\mathbb{F}_2^{n}|Enc_k^{r}(x\oplus \Delta x)\oplus Enc_k^{r}(x)=\Delta y\}|.
\end{align*}
If this value is equal to $p$, we call $(\Delta x,\Delta y)$ a $p$-probability differential of $Enc_k^r$.

Differential attack was proposed in 1991 and is one of the most commonly used cryptanalysis methods \cite{ref-journal14}. It utilizes the existence of the differentials with high probability to break block ciphers. Let $Enc_k^{t}=Enc_{k_{t}}\circ Enc_{k_{t-1}}\circ\cdots\circ Enc_{k_1}$ be the $t$-round iteration of $Enc_k$, where $1<t<r$. Namely, $Enc_k^{t}$ is a reduced cipher of $Enc_k^r$. In differential attacks the adversaries first search for a differential of $Enc_k^{t}$ having high probability, then use this differential to screen out the right subkey involved in the last $r-t$ rounds of $Enc_k^r$.

There have been some variants of differential cryptanalysis proposed, including impossible differential attack \cite{ref-proceeding19}, truncated differential attack \cite{ref-proceeding20} and boomerang attack\cite{ref-proceeding21}. These attacks all utilize the statistical no-random properties of the ciphertext differences when specifying the plaintext differences. 

Inspired by the concept of differential of block ciphers, we also define differential for Boolean functions. Let $\mathcal{C}_{u,v}$ be the set that contains all Boolean functions mapping $u$ bits to $v$ bits, where $u,v$ are arbitrary positive integers. For any $f\in\mathcal{C}_{u,v}$, any $x,x'\in\mathbb{F}_2^u$, let $f(x)=y\in\mathbb{F}_2^v$, $f(x')=y'\in\mathbb{F}_2^v$. We call $\Delta y=y'\oplus y$ an output difference and $\Delta x=x'\oplus x$ an input difference. $(\Delta x,\Delta y)$ is defined as a differential of $f$. The probability of differential $(\Delta x,\Delta y)$ is defined as
\begin{align*}
\Pr_{x\leftarrow\mathbb{F}_2^u}[\Delta x\overset{f}{\rightarrow}\Delta y]&=\Pr_{x\leftarrow\mathbb{F}_2^u}[f(x\oplus \Delta x)\oplus f(x)=\Delta y].
\end{align*}
If this value is equal to $p$, we call $(\Delta x,\Delta y)$ a $p$-probability differential of $f$. Especially, for any function $f$ in $\mathcal{C}_{u,1}$, define two sets
\begin{align*}
D_f^0&=\{\Delta x\in\mathbb{F}_2^u|f(x)\oplus f(x\oplus\Delta x)=0,\forall x\in\mathbb{F}_2^u\},\\
D_f^1&=\{\Delta x\in\mathbb{F}_2^u|f(x)\oplus f(x\oplus\Delta x)=1,\forall x\in\mathbb{F}_2^u\}.
\end{align*}
Let $D_f=D_f^0\cup D_f^1$. The vectors in $D_f$ are called complete differentials of $f$. For any vector $\Delta x\in D_f^i$ ($i\in\{0,1\}$), $(\Delta x,i)$ is obviously a 1-probability differential. For any $\Delta x\notin D_f^i$, $\Delta x$ can not form a 1-probability differential of $f$ as an input difference.

\subsection{Quantum Computing}
For any $f\in\mathcal{C}_{u,v}$, the quantum circuit realizes $f$ is equivalent to realizing the following operator
$$
U_f:\sum_{x,y}|x\rangle|y\rangle\rightarrow\sum_{x,y}|x\rangle|y\oplus f(x)\rangle.
$$
Any block cipher $Enc^r$  can
be realized efficiently through a quantum circuit. Namely, there is a quantum circuit with polynomial-complexity taking a state of plaintexts and master keys as input, and outputting the corresponding ciphertexts, realizing the unitary operator
$$
U_{Enc^r}:\sum_{k\in\mathbb{F}_2^m \atop x,y\in\mathbb{F}_2^n}|k\rangle|x\rangle|y\rangle\rightarrow\sum_{k\in\mathbb{F}_2^m \atop x,y\in\mathbb{F}_2^n}|k\rangle|x\rangle|y\oplus Enc_k^r(x)\rangle.
$$
All quantum circuits can be realized using only the gates in some universal gate set \cite{ref-book1}, such as $\{\frac{\pi}{8},CNOT,Phase,Hadamard\}$. Thus $Enc^r$ can be realized using a quantum circuit containing only polynomial universal gates. Let the total amount of quantum universal gates in this circuit be $|Enc^r|_Q$. $U_{Enc^r}$ can be integrated into the quantum circuits as presented in Figure 1.

\begin{figure}[H]
\centering
\includegraphics[width=7cm]{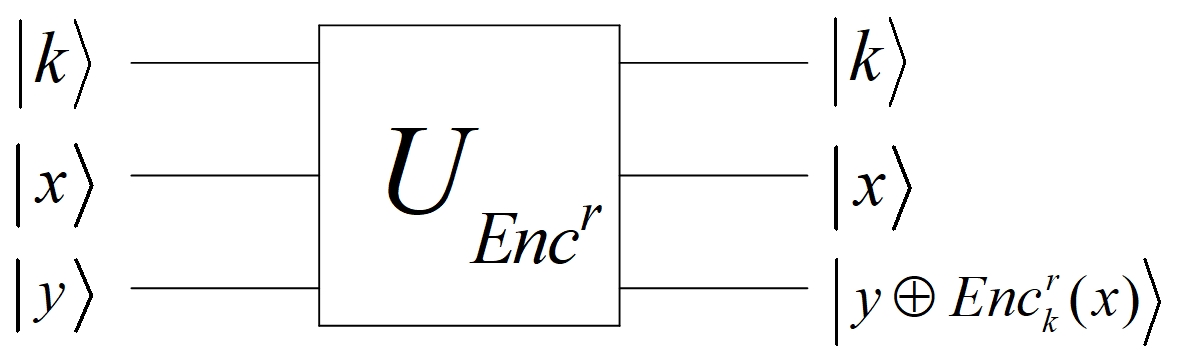}
\caption{Quantum gate $U_{Enc^r}$\label{Fig1}}
\end{figure}

There are two models proposed to describe quantum adversaries, named $Q_1$ and $Q_2$ model respectively \cite{ref-proceeding7,ref-proceeding8,ref-proceeding9}. The $Q_1$ adversaries can perform local quantum operations, but can merely make queries classically to the cryptography primitives. The $Q_2$ adversaries, in addition to classical queries and local quantum operations, can also query quantum oracles of the cryptography primitives.
$Q_2$ model is more
demanding since it is difficult to achieve quantum oracles of cryptography primitives in practice. 

\subsection{Bernstein-Vazirani Algorithm}
Bernstein-Vazirani (BV) algorithm \cite{ref-journal10} was designed for solving the problem: $s\in\{0,1\}^n$ being a secret vector, with a quantum circuit of function $f(x)=s\cdot x$, how to get the value of $s$. BV algorithm runs as follows:
\begin{enumerate}
\item Implement Hadamard transform $H^{(n+1)}$ on $|\psi_0\rangle=|0\rangle^{\otimes n}|1\rangle$, giving 
$$|\psi_1\rangle=\sum_{x\in \mathbb{F}_2^n}\frac{|x\rangle}{\sqrt{2^n}}\cdot\frac{|0\rangle-|1\rangle}{\sqrt{2}}.
$$

\item Using the quantum circuit of $f$ to get
$$|\psi_2\rangle=\sum_{x\in \mathbb{F}_2^n}\frac{(-1)^{f(x)}|x\rangle}{\sqrt{2^n}}\frac{|0\rangle-|1\rangle}{\sqrt{2}}.
$$
\item Discarding the unentangled last qubit, perform Hadamard operator $H^{(n)}$ on the remaining qubits, getting
\begin{align}
|\psi_3\rangle=\sum_{y\in \mathbb{F}_2^n}(\frac{1}{2^n}\sum_{x\in \mathbb{F}_2^n}(-1)^{f(x)+y\cdot x})|y\rangle.
\end{align}
Since $f(x)=s\cdot x$, we have
\begin{align*}
|\psi_3\rangle&=\sum_{y\in \mathbb{F}_2^n}(\frac{1}{2^n}\sum_{x\in \mathbb{F}_2^n}(-1)^{(s\oplus y)\cdot x})|y\rangle=|s\rangle.
\end{align*}
Therefore, measuring $|\psi_3\rangle$ gives the value of $s$.
\end{enumerate}

For any function $f:\mathbb{F}_2^n\rightarrow\mathbb{F}_2$ in $\mathcal{C}_{n,1}$, its Walsh transform is defined as the function
\begin{align*}
S_f:\mathbb{F}_2^n&\longrightarrow\mathbb{F}_2\\
u&\longrightarrow S_f(u)=\frac{1}{2^n}\sum_{x\in \mathbb{F}_2^n}(-1)^{f(x)+u\cdot x}.
\end{align*}
Eq.(1) shows that, if BV algorithm is run on a general function $f\in\mathcal{C}_{n,1}$, the final quantum state without measured will be 
$$\sum_{y\in \mathbb{F}_2^n}S_f(y)|y\rangle,$$ 
where $S_f(\cdot)$ denotes the Walsh transform of $f$. When this state is measured, the probability of any $y\in\mathbb{F}_2^n$ being output is $S_f(y)^2$. Thus BV algorithm run on $f$ must output $y$ such that $S_f(y)\neq0$. 

Figure 3 shows the circuit of BV algorithm. BV algorithm needs totally $2n+1+|f|_Q$ universal gates. The corresponding quantum circuit needs $n+1$ qubits. 

\begin{figure}[H]
\centering
\includegraphics[width=10cm]{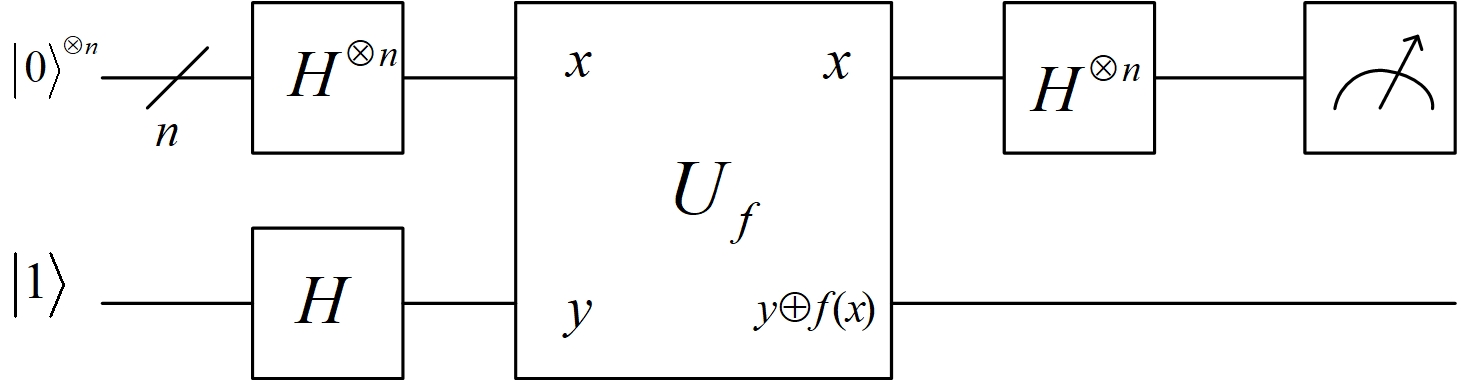}
\caption{Quantum circuit of BV algorithm\label{Fig2}}
\end{figure}

Utilizing the fact that BV algorithm always outputs the vectors in the support set of the Walsh transform, Li et al. showed a quantum algorithm that is used for finding differentials having high probability \cite{ref-journal15}.

\begin{algorithm}[H]
\caption*{{\large \textbf{Algorithm 1}}}
\hangafter 1
\hangindent 4.1em
\textbf{Input$\,$}: $\,\,$the quantum circuit realizing $f\in\mathcal{C}_{n,1}$, a polynomial $q(n)$ of $n$.

\hangafter 1
\hangindent 4.7em
\textbf{Output}: a differential of function $f$.

\begin{algorithmic}[1]
\State \parbox[t]{\dimexpr\textwidth-\leftmargin-\labelsep-\labelwidth}{Define a set $W:=\Phi$;\strut}
    \For{$l=1,2,\cdots,q(n)$}
    \State \parbox[t]{\dimexpr\textwidth-\leftmargin-\labelsep-\labelwidth}{Apply BV algorithm to $f$, obtaining an vector $u$ with $n$ bits such that $S_f(u)\neq0$;\strut}

       \State \parbox[t]{\dimexpr\textwidth-\leftmargin-\labelsep-\labelwidth}{Let $W=W\cup \{u\}$;\strut}
\EndFor
\State \parbox[t]{\dimexpr\textwidth-\leftmargin-\labelsep-\labelwidth}{Solve the equation $\{x\cdot u=i|u\in W\}$, getting the solution set $Z^i$ for both $i=0,1$;\strut}
\If{ $Z^0\cup Z^1\subseteq\{\vec{0}\}$}  \State \parbox[t]{\dimexpr\textwidth-\leftmargin-\labelsep-\labelwidth}{Output ``No'';\strut}
 \Else \State \parbox[t]{\dimexpr\textwidth-\leftmargin-\labelsep-\labelwidth}{Output $Z^0$ and $Z^1$;\strut}
 \EndIf
\end{algorithmic}
\end{algorithm}

For any $f\in\mathcal{C}_{n,1}$, let
\begin{align*}
\gamma_f&=\frac{1}{2^n} \max_{\substack{\Delta x\in \mathbb{F}_2^n\\ \Delta x\notin D_f}}\max_{i\in \{0,1\}}|\{x\in \mathbb{F}_2^n|f(x\oplus \Delta x)\oplus f(x)=i\}|\\
&=\max_{\substack{\Delta x\in \mathbb{F}_2^n\\ \Delta x\notin D_f}}\max_{i\in \{0,1\}}\Pr_{x\leftarrow\mathbb{F}_2^n}[\Delta x\overset{f}{\rightarrow}i].
\end{align*}
It is easy to derive that $\gamma_f<1$. This parameter is the maximum probability of differentials of $f$ which are not 1-probability differentials. The authors in \cite{ref-journal11,ref-journal15} has proven the following theorems, which illustrate the soundness of the above algorithm.
\begin{Theorem}[\cite{ref-journal15}] If Algorithm 1 outputs two sets $Z^0$ and $Z^1$ when applied to a function $f\in\mathcal{C}_{n,1}$, then for any $\Delta x\in Z^i$ ($i=0,1$), any $\epsilon$ satisfying $0<\epsilon<1$, it holds that 
\begin{equation}
{\rm Pr}\Big[1-\frac{|\{x\in \mathbb{F}_2^n|f(x\oplus \Delta x)+f(x)=i\}|}{2^n}<\epsilon\Big]>1-e^{-2q(n)\epsilon^2}.
\end{equation}
\end{Theorem}

\begin{Theorem}[\cite{ref-journal11}] Suppose $f\in\mathcal{C}_{n,1}$ and there is a constant $a_0$ such that $\gamma_f\leq a_0<1$. If Algorithm 1 outputs two sets $Z^0$ and $Z^1$ when applied to $f$ with $q(n)=n$, then for any vector $\Delta x\notin D_f^i$ ($i=0,1$), we have
$$
\Pr[\Delta x\in Z^i]\leq a_0^{n}.
$$
\end{Theorem}

Theorem 1 demonstrates that, for any vector $\Delta x\in Z^i$ ($i=0,1$), the differential probability of $(\Delta x, i)$ to $f$ is greater than $1-\epsilon$ with a probability greater $1-e^{-2q(n)\epsilon^2}$.

\section{Quantum Truncated Differential Attack}

Knudsen introduced truncated differential attack in 1994 \cite{ref-proceeding20}. This cryptanalytic method has been widely applied to attack symmetric ciphers \cite{ref-proceeding22,ref-proceeding23}. In initial version of differential attacks, the adversaries utilize full differences of plaintexts and ciphertexts, while in truncated differential attacks the adversaries consider differences which are just partially determined. The adversaries only predict some bits of the differentials rather than the entire differentials.

We still consider the block cipher $Enc_k^r$ with key space $\mathbb{F}_2^m$. A truncated differential $(\overline{\Delta} x,\overline{\Delta} y)$ of $Enc_k^r$ is a pair of vectors satisfying $\overline{\Delta} x, \overline{\Delta} y\in\{*,0,1\}^n$. where $*$ denotes an undetermined bit. Let $\overline{\Delta} x=(\overline{\Delta}x_1,\cdots,\overline{\Delta}x_n)$, $\overline{\Delta} y=(\overline{\Delta}y_1,\cdots,\overline{\Delta}y_n)$. then $\overline{\Delta}x_i, \overline{\Delta}y_i\in\{*,0,1\}$. The bits of $\overline{\Delta} x$ ($\overline{\Delta} y$) taking a value of 0 or 1 are defined as predicted bits, while the bits with a value of $*$ are defined as unpredicted bits. 

A truncated difference is equivalent to a set of complete differences. Define 
$$
\Omega_{\overline{\Delta}x}=\big\{\Delta x=(\Delta x_1,\cdots,\Delta x_n)\in\mathbb{F}_2^n\big|\Delta x_i=\overline{\Delta}x_i \text{ if } \overline{\Delta}x_i\neq *, i\in\{1,2,\cdots,n\}\big\},
$$
$$
\Omega_{\overline{\Delta}y}=\Big\{\Delta y=(\Delta y_1,\cdots,\Delta y_n)\in\mathbb{F}_2^n\Big|\Delta y_i=\overline{\Delta}y_i \text{ if } \overline{\Delta}y_i\neq *, i\in\{1,2,\cdots,n\}\Big\},
$$
then truncated differences $\overline{\Delta}x$ and $\overline{\Delta} y$ are equivalent to $\Omega_{\overline{\Delta}x}$ and $\Omega_{\overline{\Delta}y}$ respectively. If a complete input difference $\Delta x$ is in $\Omega_{\overline{\Delta}x}$, that is, $\Delta x_j=\overline{\Delta}x_j$ for all $j\in\{1,\cdots,n\}$ such that $\overline{\Delta}x_j\neq *$, we say that $\Delta x$ matches the truncated difference $\overline{\Delta}x$, and denote this case as $\Delta x \sim \overline{\Delta}x$. Likewise, $\Delta y \sim \overline{\Delta}y$ means that $\Delta y$ matches the truncated difference $\overline{\Delta}y$. 

The conditional probability
\begin{align*}
\Pr_{x\leftarrow\mathbb{F}_2^n}[\overline{\Delta}x\overset{Enc_k^{r}}{\rightarrow}\overline{\Delta}y]&=\Pr_{x\leftarrow\mathbb{F}_2^n}[Enc_k^{r}(x\oplus \Delta x)\oplus Enc_k^{r}(x)\sim \overline{\Delta}y | \Delta x\sim\overline{\Delta}x ]\\
&=\Pr_{x\leftarrow\mathbb{F}_2^n}[Enc_k^{r}(x\oplus \Delta x)\oplus Enc_k^{r}(x)\in \Omega_{\overline{\Delta}y}|\Delta x\in\Omega_{\overline{\Delta}x}]
\end{align*}
is defined as the probability of $(\overline{\Delta}x,\overline{\Delta}y)$. If $p$ is equal to the probability of $(\overline{\Delta}x,\overline{\Delta}y)$, we call $(\overline{\Delta}x,\overline{\Delta}y)$ a $p$-probability truncated differential of $Enc_k^r$.

Let $Enc^t$ ($1<t<r$) be a reduced cipher of $Enc^r$. In a truncated differential attack, the adversaries first search for a truncated differential of $Enc^t$ which has high probability, then use this truncated differential, denoted as $(\overline{\Delta}x,\overline{\Delta}y)$ to recover the subkeys involved in the last $r-t$ rounds. In detail, the adversaries fix the plaintext difference $\overline{\Delta} x$, then use $2M$ pairs of plaintexts, whose differences match $\overline{\Delta} x$, to make encryption queries and get $2M$ pairs of corresponding ciphertexts. Subsequently, for each possible candidate subkey of the last $r-t$ rounds, the adversaries use it to decrypt $r-t$ rounds to obtain $M$ output differences of $Enc_k^t$, in the meantime calculate the amount of the differences that match $\overline{\Delta}y$. Finally, the right subkey is the subkey having the maximum count.

The amount of plaintext pairs needed in a such counting scheme and the success probability of getting right key are determined by the ratio of signal to noise \cite{ref-journal14}, whose definition is:
$$
S/N=\frac{L\times p}{\alpha\times \lambda}.
$$
Here $L$ denotes the total amount of possible subkeys involved in the last $r-t$ rounds, $p$ denotes the probability of $(\overline{\Delta}x,\overline{\Delta}y)$, $\alpha$ denotes the average count that every plaintext pair contributes and $\lambda$ denotes the proportion of pairs not discarded in the preprocessing procedure. We don't consider any pre-discarding process here, therefore set $\lambda=1$. A truncated differential attack can succeed only when $S/N>1$. Thus the adversaries should use a truncated differential that makes the ratio of signal to noise greater than 1. The greater $S/N$ is, the easier it is to singled out the right subkey. 

In the following we propose a quantum algorithm which is used for finding truncated differentials. In classical truncated differential attack, because the adversaries do not know the value of $k$ of the reduced cipher $Enc_k^t$, they actually have to find a truncated differential whose probability is high no matter what value the key $k$ takes. Due to this, our quantum algorithm is designed to search truncated differentials having high probability for a large proportion of keys in $\mathbb{F}_2^m$. Specifically, by choosing an polynomial $\tau(n)$, the adversaries can force our quantum algorithm only output truncated differentials that, for more than $(1-\frac{1}{\tau(n)})$ proportion of keys in $\mathbb{F}_2^m$, has high probability. We first present the algorithm and than analyze its effectiveness and complexity.

\subsection{Finding Truncated Differentials via BV algorithm}

Given a reduced block cipher $Enc_k^t$, let $Enc^t_{k}(x)=(Enc^t_{k}[1](x),\cdots,$ $Enc^t_{k}[n](x))$. That is, $Enc_k^t[j]$ denotes the $j$-th component function of $Enc_k^t$. An intuitive way to find high-probability truncated differentials is to implement Algorithm 1 on every $Enc_k^t[j]$. If Algorithm 1 finds differentials of several component functions that all have high probability and have a common input difference, then we can derive a truncated differential of $Enc_k^t$ having high probability. However, running Algorithm 1 to $Enc_k^t[j]$ needs quantum queries of $Enc_k^t$. It is impossible to achieve even under $Q_2$ model since $Enc_k^t$ is a reduced cipher instead of the complete cipher $Enc_k^r$. In original differential attack, the adversaries are also not able to query the reduced version. They thus analyze the detailed constructions of the cipher and search for truncated differentials whose probability are high no matter what value the key takes. Inspired by this idea, we consider searching for the truncated differentials with high probability for most keys.

Since all constructions of the cipher $Enc_k^t$, except for the private key $k$, are public, the function
\begin{align*}
Enc^t:\{0,1\}^{n}\times\{0,1\}^m&\longrightarrow \{0,1\}^n\\
(\,\,\,x\,\,\,\,\,\,,\,\,\,\, k\,\,\,)\quad&\longrightarrow Enc^t_{k}(x)
\end{align*}
take the key as input and is known and determined to the adversaries. Thus the adversaries have access to the quantum circuit of the unitary operator $$U_{Enc^t}:|x\rangle|k\rangle|y\rangle\rightarrow|x\rangle|k\rangle|y\oplus Enc^t(x,k)\rangle=|x\rangle|k\rangle|y\oplus Enc^t_k(x)\rangle.$$
Let $|Enc^t|_Q$ be the amount of quantum universal gates of this circuit. The adversaries also have the quantum circuit of every component function 
\begin{align*}
Enc^t[j]:\{0,1\}^{n}\times\{0,1\}^m&\longrightarrow \{0,1\}^n\\
(\,\,\,x\,\,\,\,\,\,,\,\,\,\, k\,\,\,)\quad&\longrightarrow Enc^t_{k}[j](x).
\end{align*}
The corresponding amount of gates is $|Enc^t[j]|_Q$ ($j=1,\cdots,n$). The adversaries have the quantum circuits of $Enc^t[j]$'s, so they can run Algorithm 1 on $Enc^t[j]$'s without quantum queries. The adversaries can run Algorithm 1 to get the differentials of high probability of every $Enc^t[j]$, then by taking a common input difference of part component functions as the input difference, they can obtain a truncated differential having high probability.
According to this idea, we propose Algorithm 2 that finds truncated differentials of block ciphers.

\begin{algorithm}[H]
\caption*{{\large \textbf{Algorithm 2}}}
\hangafter 1
\hangindent 4.1em
\textbf{Input$\,$}: $\,\,$The quantum circuit of $Enc^t$, a polynomial $\tau(n)$ and a constant $\sigma$ ($0<\sigma<1$) chosen by the adversaries.

\hangafter 1
\hangindent 4.7em
\textbf{Output}: a high-probability truncated differential of $Enc^t$.

\begin{algorithmic}[1]
\State \parbox[t]{\dimexpr\textwidth-\leftmargin-\labelsep-\labelwidth}{Let $q(n)=\frac{1}{2(1-\sigma)^2}\tau(n)^2n^3$;\strut} 
\State \parbox[t]{\dimexpr\textwidth-\leftmargin-\labelsep-\labelwidth}{Define a set $W:=\Phi$;\strut}
\For{$j=1,2,\cdots,n$}
 \For{$l=1,\cdots,q(n)$}
  \State \parbox[t]{\dimexpr\textwidth-\leftmargin-\labelsep-\labelwidth}{Apply BV algorithm to $Enc^t[j]$ to get an output $u=(u_1,\cdots,u_n,u_{n+1},\cdots,u_{n+m})$;\strut} 
  \State \parbox[t]{\dimexpr\textwidth-\leftmargin-\labelsep-\labelwidth}{let $W=W\cup\{(u_1,\cdots,u_n)\}$;\strut} 
 \EndFor
 \State \parbox[t]{\dimexpr\textwidth-\leftmargin-\labelsep-\labelwidth}{Solve the equation $\{x\cdot u=i_j|u\in W\}$, getting the solution sets $Z_j^{i_j}$ for $i_j=0,1$ respectively;\strut} 
 \State \parbox[t]{\dimexpr\textwidth-\leftmargin-\labelsep-\labelwidth}{ Let $Z_j=Z_j^0\cup Z_j^1$; \strut} 
 \State \parbox[t]{\dimexpr\textwidth-\leftmargin-\labelsep-\labelwidth}{Let $\overline{Z}_j=\{(a,ij)\in\mathbb{F}_2^n\times\mathbb{F}_2|a\in Z_j^{i_j},=0,1\}$; \strut} 
 \State \parbox[t]{\dimexpr\textwidth-\leftmargin-\labelsep-\labelwidth}{Let $W=\Phi$;\strut} 
\EndFor
\For{$d=1,2,\cdots,n$}
 \If{$S/N=2^d\sigma>1$}
 \If{there are $d$ different subscripts $j_1,\cdots,j_d$ s.t. $Z_{j_1}\cap\cdots\cap Z_{j_d}\supsetneq\{\vec{0}\}$}
 \State \parbox[t]{\dimexpr\textwidth-\leftmargin-\labelsep-\labelwidth}{Choose at random a vector $a\in Z_{j_1}\cap\cdots\cap Z_{j_d}$, and for $j=1,\cdots,n$, let
 $$
b_j=\left\{\begin{array}{cc}
i_j, & \,\,j\in\{j_1,\cdots,j_d\}\\
*, &  \,\,j\notin\{j_1,\cdots,j_d\},\\
\end{array}
\right.
$$\strut}
\algstore{bkbreak}
\end{algorithmic}
\end{algorithm}

\begin{algorithm}
\begin{algorithmic}[1]
\algrestore{bkbreak}
\State \parbox[t]{\dimexpr\textwidth-\leftmargin-\labelsep-\labelwidth}{where $i_j$ denotes the bit 
appended to $a$ in the set $\overline{Z}_j$, i.e., $(a,i_j)\in \overline{Z}_j$;\strut} 
\State \parbox[t]{\dimexpr\textwidth-\leftmargin-\labelsep-\labelwidth}{ Let $b=(b_1,\cdots,b_n)$ and return $(a,b)$;\strut} 
 \EndIf
 \EndIf
\EndFor
\State \parbox[t]{\dimexpr\textwidth-\leftmargin-\labelsep-\labelwidth}{Return ``No'';\strut} 
\end{algorithmic}
\end{algorithm}

Steps 1-12 of Algorithm 2 are for finding high-probability differentials of $Enc^t[j]$ for every $j=1,\cdots,n$. The purpose of steps 13-20 is to choose a difference which is a common input difference of as many $Enc^t[j]$ as possible. Algorithm 2 outputs a truncated differential $(a,b)$ of $Enc^t$. The symbol ``$*$'' in $b$ means the corresponding bits are unpredicted. In a quantum truncated differential attack, the adversaries first choose a polynomial $\tau(n)$ and a constant $\sigma$ ($0<\sigma<1$), then implement Algorithm 2 to get an output $(a,b)$. According to Theorem 3 which will be proved in section 3.2, the differential probability of $(a,b)$ is greater than $\sigma$ for more than $(1-\frac{1}{\tau(n)})$ proportion of keys in $\mathbb{F}_2^m$ with a overwhelming probability.

To demonstrate the feasibility of the output truncated differential $(a,b)$, it is necessary to compute the ratio of signal to noise $S/N$. To this end, we first calculate the parameter $\alpha$, which is equal to the average count that every plaintext pair contributes. There are $d$ bits of the difference $b$ predicted, therefore a total of $2^{n-d}$ output differences matching the truncated difference $b$. In counting process, the ciphertexts of a fixed pair of plaintexts are decrypted by $L$ candidate subkeys. The resulting $L$ output differences can be viewed as random vectors. Therefore, every plaintext pair contributes $$\alpha=\frac{2^{n-d}}{2^n}\times L=\frac{L}{2^d}$$ counts on average. Then  
$$
N/S\geq\frac{L\times\sigma}{\frac{L}{2^d}\times1}=2^d\sigma>1.
$$
This value is greater than 1 because of the condition $2^d\sigma>1$ in the step 14 of Algorithm 2. After getting the output $(a,b)$, the adversaries can utilize it to find the right subkey involved in the last $r-t$ rounds just like the traditional truncated differential attack. This attack should work for at least $(1-\frac{1}{\tau(n)})$ proportion of keys in $\mathbb{F}_2^m$. Even if ``No'' is output, the adversaries can adjust the polynomial $\tau(n)$ and $\sigma$ to increase the success probability.

\subsection{Analysis}

We analyze the correctness and efficiency of Algorithm 2. Theorem 3 indicates the correctness of Algorithm 2. 
\begin{Theorem}Suppose Algorithm 2 outputs $(a,b)$, then with a overwhelming probability, there is a subset $S\subseteq\mathbb{F}_2^m$ satisfying that $|S|/|\mathbb{F}_2^m|>1-\frac{1}{\tau(n)}$, and for every key $k\in S$, 
$$
\frac{|\{x\in \mathbb{F}_2^n|Enc^t_k(x\oplus a)+Enc^t_k(x)\sim b\}|}{2^n}>\sigma.
$$
Namely, the differential probability of $(a,b)$ is greater than $\sigma$ for more than $(1-\frac{1}{\tau(n)})$ proportion of keys in $\mathbb{F}_2^m$.
\end{Theorem}

\begin{proof}[Proof] $b$ has $d$ predicted bits, whose subscripts are $j_1,\cdots,j_d$. Appending $m$ zeros after the vector $a$ gives a $(n+m)$-bit vector $(a\|0,\cdots,0)$. Since $a\cdot(u_1,\cdots,u_n)=0$, it holds that
$$(a\|0,\cdots,0)\cdot(u_1,\cdots,u_n,u_{n+1},\cdots,u_{n+m})=0.$$
We can view the $(n+m)$-bit vector $(a\|0,\cdots,0)$ as the output of Algorithm 2 when it is applied to $Enc^t[j]$ for all $j\in\{j_1,j_2,\cdots,j_d\}$. Due to Theorem 1, the probability that, 
\begin{equation*}
\frac{|\{z\in \mathbb{F}_2^{n+m}|Enc^t[j](\,z\oplus(a\|0,\cdots,0)\,)\oplus Enc^t[j](z)=b_{j})\}|}{2^{n+m}}>1-\epsilon,\,\,
\end{equation*}
holds $\forall j\in\{j_1,j_2,\cdots,j_d\}$ is greater than $(1-e^{-2q(n)\epsilon^2})^d$. If the above inequality holds, then the number of $z$ that satisfies
\begin{equation}
Enc^t[j]\big(z\oplus(a\|0,\cdots,0)\big)\oplus Enc^t[j](z)=b_{j}
\end{equation}
for both $j=j_1$ and $j=j_2$ is greater than $2^{n+m}[2(1-\epsilon)-1]=2^{n+m}(1-2\epsilon)$. Likewise, the number of $z$ that satisfies Eq.$(3)$ for all $j=j_1,j_2,j_3$ is greater than $2^{n+m}(1-3\epsilon)$. By induction, the number of $z$ that satisfies Eq.$(3)$ for all $j\in\{j_1,j_2,\cdots,j_d\}$ is more than $2^{n+m}(1-d\epsilon)$. Therefore, the probability that
\begin{equation*}
\frac{|\{z\in \mathbb{F}_2^{n+m}|Enc^t(\,z\oplus(a\|0,\cdots,0)\,)\oplus Enc^t(z)\sim b)\}|}{2^{n+m}}>1-d\epsilon.
\end{equation*}
holds is greater than $(1-e^{-2q(n)\epsilon^2})^d$, which is equivalent to
\begin{equation}
\frac{|\{(x,k)\in \mathbb{F}_2^n\times \mathbb{F}_2^m|Enc^t_k(x\oplus a)\oplus Enc^t_k(x)\sim b\}|}{2^{n+m}}>1-d\epsilon.
\end{equation}
Let
$$
Z(k)=\frac{|\{x\in \mathbb{F}_2^n|Enc^t_k(x\oplus a)+Enc^t_k(x)\sim b\}|}{2^n}.
$$
Eq.$(4)$ means that $\mathbb{E}_k[Z(k)]>1-d\epsilon$. Here $\mathbb{E}_k[Z(k)]$ is the statistical expectation of $Z(k)$ while the variable $k$ follows the uniform distribution of $\mathbb{F}_2^m$. Therefore, when Eq.$(4)$ holds, we have
$$
\Pr_k\big[\,Z(k)>1-\tau(n)d\epsilon\,\big]>1-\frac{1}{\tau(n)}
$$
for any polynomial $\tau(n)$. This is because, if not, then ${\rm Pr}_{k\leftarrow\mathbb{F}_2^m}[1-Z(k)\geq \tau(n)d\epsilon]\geq\frac{1}{\tau(n)}$, which means
\begin{align*}
&\mathbb{E}_k[Z(k)]\\
=&1-\mathbb{E}_k[1-Z(k)]\\
\leq&1-\frac{1}{\tau(n)}\cdot \tau(n)d\epsilon\\
=&1-d\epsilon.
\end{align*}
This causes contradiction. Thus, as long as Eq.$(4)$ holds, the proportion of the keys satisfying $Z(k)>1-\tau(n)d\epsilon$ in $\mathbb{F}_2^m$ must be greater than $(1-\frac{1}{\tau(n)})$. Let $S$ be the set of all such keys. We have $|S|/|\mathbb{F}_2^m|>1-\frac{1}{\tau(n)}$, and for every $k\in S$, 
$$
Z(k)=\frac{|\{x\in \mathbb{F}_2^n|Enc^t_k(x\oplus a)+Enc^t_k(x)\sim b\}|}{2^n}>1-\tau(n)d\epsilon.
$$
Let $\epsilon=\frac{1-\sigma}{\tau(n)d}$. Since $q(n)=\frac{1}{2(1-\sigma)^2}\tau(n)^2n^3$, the probability that Eq.$(4)$ holds is larger than $1-ne^{-n}$. Therefore, with a overwhelming probability, there is a subset $S\subseteq\mathbb{F}_2^m$ satisfying that $|S|/|\mathbb{F}_2^m|>1-\frac{1}{\tau(n)}$, and for every $k\in S$,
\begin{align*}
\frac{|\{x\in \mathbb{F}_2^n|Enc^t_k(x\oplus a)+Enc^t_k(x)\sim b\}|}{2^n}>1-\tau(n)d\epsilon=\sigma,
\end{align*}
which means that the differential probability of $(a,b)$ is greater than $\sigma$ for more than $(1-\frac{1}{\tau(n)})$ proportion of keys in $\mathbb{F}_2^m$
\end{proof}

When implementing truncated differential attack, the adversaries first run Algorithm 2 to get $(a,b)$. Then with a overwhelming probability, for at least $(1-\frac{1}{\tau(n)})$ proportion of keys in $\mathbb{F}_2^m$ the probability of $(a,b)$ is greater than $\sigma$. Then the adversaries can use $(a,b)$ to determine the subkey of the last $r-t$ rounds like in traditional truncated differential attack. This attack  will work for at least $(1-\frac{1}{\tau(n)})$ proportion of keys in $\mathbb{F}_2^m$. The amount of plaintext pairs required in counting process is determined by the value of $S/N$. Based on experimental observations, about 20 to 40 appearances of right plaintext pairs are enough \cite{ref-journal14}. Therefore, about $\frac{40}\sigma$ plaintext pairs are sufficient.

For analyzing the complexity, we first calculate the amounts of universal gates and qubits required, then estimate the complexity of involved classical computing.

In Algorithm 2, BV algorithm is performed on each $Enc^t[j]$ for $q(n)$ times ($j\in\{1,2,\cdots,n\}$). Every call requires $2(m+n)+1+|Enc^t[j]|_Q$ quantum universal gates. Therefore, Algorithm 2 requires
\begin{align*}
&q(n)\sum_{j=1}^n\big[\,2(m+n)+1+|Enc^t[j]|_Q\big]\\
=&q(n)\big[\,(2m+1)n+2n^2+\sum_{j=1}^n|Enc^t[j]|_Q\big]\\
=&\frac{1}{2(1-\sigma)^2}\tau(n)^2n^3\big[\,2n^2+(2m+1)n+|Enc^t|_Q\big],
\end{align*}
universal gates. This number is a polynomial of $n$ and $m$.

The involved classical computing is to solve the linear system $\{x\cdot u=i_j|u\in W\}$ for each $j=1,2,\cdots,n$ and $i_j=0,1$. The adversaries need to solve a total of $2n$ systems, every system has $q(n)$ equations and $n$ unknowns. Therefore, the classical complexity of this part is $O(2q(n)n^3)=O(\frac{1}{(1-\sigma)^2}\tau(n)^2n^6)$.

Applying BV algorithm to every $Enc^t[j]$ requires $m+n+1$ qubits. Because BV algorithm is executed sequentially, the adversaries can reuse these qubits. Thus a total of $n+m+1$ qubits is enough for Algorithm 2.

\section{Quantum Boomerang Attack}
Since proposed in 1999, boomerang attack \cite{ref-proceeding21} has been widely used as a cryptanalysis method. The idea of boomerang cryptanalysis is to connect two differential paths with high probability so that the adversaries can attack more rounds. This attack was proposed due to the reality that, when constructing differential characteristics of block ciphers, the probability of differential rapidly decreases as the round number increases. It works for the case where it is difficult to find a $(t_1+t_2)$-round differential characteristic of some block ciphers that has high probability, while it is possible to find $t_1$-round and $t_2$-round differential characteristics having high probability.
\begin{figure}[H]
\centering
\includegraphics[width=8cm]{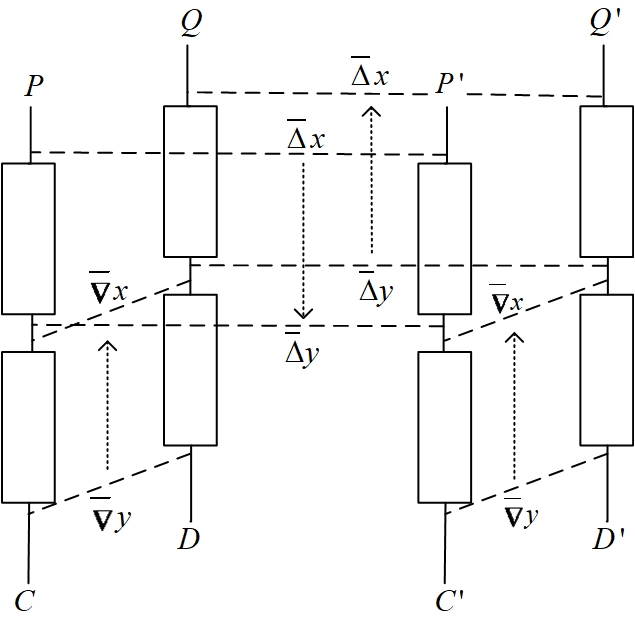}
\caption{Boomerang attack\label{Fig3}}
\end{figure}
Suppose $Enc_k^r(x)=Enc_k^{t_2}\circ Enc_k^{t_1}(x)$, where $t_1+t_2=r$, $Enc_k^{t_1}$ has a $p_1$-probability truncated differential $(\overline{\Delta} x,\overline{\Delta} y)$ and the inverse function ${Enc_k^{t_2}}^{-1}$ of $Enc_k^{t_2}$ has a $p_2$-probability truncated differential $(\overline{\nabla} x,\overline{\nabla} y)$. As shown in Figure 3, $P,P',Q,Q'$ are four plaintexts and the corresponding ciphertexts under $Enc_k^r$ are $C,C',D,D'$ respectively. $P,P'$ are said to satisfy the differential $(\overline{\Delta} x,\overline{\Delta} y)$ of $Enc_k^{t_1}$, if $P\oplus P'\sim \overline{\Delta} x$ and $Enc_k^{t_1}(P)\oplus Enc_k^{t_1}(P')\sim \overline{\Delta} y$. If both $P,P'$ and $Q,Q'$ satisfy the differential $(\overline{\Delta} x,\overline{\Delta} y)$ of $Enc_k^{t_1}$, and both $C,D$ and $C',D'$ satisfy the differential $(\overline{\nabla} x,\overline{\nabla} y)$ of ${Enc_k^{t_2}}^{-1}$, then $(P,P',Q,Q')$ is called a right quadruple. Such two differentials are called a boomerang distinguisher of $Enc_k^r$. Quadruples can be generated via the following method: (1) choose two plaintexts $(P,P')$ satisfying $P\oplus P\sim \overline{\Delta} x$ and denote the corresponding ciphertexts as $(C,C')$; (2) compute $D=C\oplus\overline{\nabla} y$ and $D'=C'\oplus \overline{\nabla} y$, and decrypt $D,D'$ to get the corresponding plaintexts $Q,Q'$; (3) test whether it holds that $Q\oplus Q'\sim\overline{\Delta} x$.

The probability of such $(P,P')$ satisfying differential $(\overline{\Delta} x,\overline{\Delta} y)$ is $p_1$. The probability of $(C,D)$ satisfying the differential $(\overline{\nabla} x,\overline{\nabla} y)$ is $p_2$. The probability of $(C',D')$ satisfying the differential $(\overline{\nabla} x,\overline{\nabla} y)$ is also $p_2$. Under the above three conditions, it naturally holds that $Enc_k^{t_1}(Q)\oplus Enc_k^{t_1}(Q')\sim\overline{\Delta} y$, so that the probability of $Q\oplus Q'\sim \overline{\Delta} x$ is $p_1$. In summary, the probability of a quadruple generated by the above method being a right quadruple is $(p_1p_2)^2$. For a random permutation, this probability is $2^{-d}$, where $d$ is the number of determined bits of $\overline{\Delta} x$. If $(p_1p_2)^2>2^{-d}$, then the block cipher can be distinguished from a random permutation through data analysis. A boomerang distinguisher can be used to search for the subkey involved in the last several rounds of the attacked block cipher.

The key of the boomerange attack is to find the boomerang distinguisher, namely, a $t_1$-round truncated differential of $Enc_k^{t_1}$ and a $t_2$-round truncated differential of ${Enc_k^{t_2}}^{-1}$ that have high probability. Thus the essence of boomerange attack is to find two truncated differentials that have high probability, which can be done by Algorithm 2. According to these analysis, we propose Algorithm 3 for finding boomerange distinguishers.

\begin{algorithm}[H]
\caption*{{\large \textbf{Algorithm 3}}}
\hangafter 1
\hangindent 4.1em
\textbf{Input$\,$}: $\,\,$The quantum circuit of $Enc^r$, a polynomial $\tau(n)$ and a constant $\sigma$ ($0<\sigma<1$) chosen by the adversaries.

\hangafter 1
\hangindent 4.7em
\textbf{Output}: a boomerange distinguisher of $Enc^r$.

\begin{algorithmic}[1]
\For{$t_1=1,2,\cdots,r$}
    \State \parbox[t]{\dimexpr\textwidth-\leftmargin-\labelsep-\labelwidth}{Let $t_2=r-t_1$;\strut} 
  \State \parbox[t]{\dimexpr\textwidth-\leftmargin-\labelsep-\labelwidth}{Run Algorithm 2 on $Enc^{t_1}$ with the parameter $\sigma$ and polynomial $\tau(n)$ to get an output $(a_1,b_1)$ or ``No'';\strut} 
  \State \parbox[t]{\dimexpr\textwidth-\leftmargin-\labelsep-\labelwidth}{Run Algorithm 2 on ${Enc^{t_2}}^{-1}$ with the parameter $\sigma$ and polynomial $\tau(n)$ to get an output $(a_2,b_2)$ or ``No'';\strut}
  \If{both outputs of Algorithm 2 are not ``No'' }
 \State \parbox[t]{\dimexpr\textwidth-\leftmargin-\labelsep-\labelwidth}{Return $\{(a_1,b_1),(a_2,b_2),t_1,t_2\}$;\strut}   
 \EndIf
 \EndFor 
\end{algorithmic}
\end{algorithm}

Step 3 of Algorithm 3 is for finding truncated differentials of $Enc^{t_1}$, of which the probability is larger than $\sigma$ for at least $1-\frac{1}{\tau(n)}$ proportion of keys in $\mathbb{F}_2^m$. Step 4 of Algorithm 3 is for finding truncated differentials of ${Enc^{t_2}}^{-1}$, of which the probability is also larger than $\sigma$ for at least $1-\frac{2}{\tau(n)}$ proportion of keys in $\mathbb{F}_2^m$.
These two differentials form a boomerange distinguisher $\{(a_1,b_1),(a_2,b_2)\}$ of $Enc_k^r$. The probability of the corresponding  quadruple is $\sigma^4$ for more than $1-\frac{2}{\tau(n)}$ proportion of keys in $\mathbb{F}_2^m$. 

Algorithm 3 calls Algorithm 2 on $Enc^{t_1}$ and ${Enc^{t_2}}^{-1}$ respectively for all $t_1=1,2,\cdots,n$ and $t_2=r-t_1$. Therefore, it requires
\begin{align*}
&\sum_{t_1=1}^{r}\frac{1}{2(1-\sigma)^2}\tau(n)^2n^3\big[\,4n^2+(4m+2)n+(|Enc^{t_1}|_Q+|{Enc^{t_2}}^{-1}|_Q)\big]\\
=&\sum_{t_1=1}^{r}\frac{1}{2(1-\sigma)^2}\tau(n)^2n^3\big[\,4n^2+(4m+2)n+(|Enc^{t_1}|_Q+|{Enc^{t_2}}|_Q)\big]\\
=&\sum_{t_1=1}^{r}\frac{1}{2(1-\sigma)^2}\tau(n)^2n^3\big[\,4n^2+(4m+2)n+|Enc^{r}|_Q\big]\\
=&\frac{r}{2(1-\sigma)^2}\tau(n)^2n^3\big[\,4n^2+(4m+2)n+|Enc^{r}|_Q\big]
\end{align*}
universal gates. This number is a polynomial of $n$ and $m$.

Since Algorithm 2 is called sequentially in Algorithm 3, the adversaries can reuse the qubits. Thus $m+n+1$ qubits is enough for Algorithm 3.

\section{Results and Discussion}

In this work, we further explore the 
superior computing capacity of quantum computing when applied to the field of cryptanalysis. We use BV algorithm to enhance two variants of differential cryptanalysis: truncated differential cryptanalysis and boomerang cryptanalysis. We constructed two quantum algorithms which can find truncated differentials and boomerang distinguishers of block ciphers respectively. We prove that, with a overwhelming probability, the truncated differentials or boomerang distinguishers found by our algorithms have  high probability for the most keys in key space. 

The complexity of our algorithms is  polynomial level, and the adversaries can realize them in Q1 model. Compared to many proposed quantum attack algorithms \cite{ref-proceeding3,ref-proceeding4,ref-journal2,ref-proceeding5,ref-journal3,ref-journal4} which demand quantum queries, our algorithms are more practical for realizing. Classical automatic tools for searching truncated differentials with high probability or boomerang distinguishers was unable to consider all details of S-boxes when the S-boxes are not  small-scale. For example, facing the widely used 8-bit S-boxes, the classical searching tools can only work for extremely few rounds. By comparison, our algorithms fully utilize the strengths of quantum computing to make up this shortcoming. Their quantum circuit strictly compute the S-boxes when performing the operator $U_{Enc^t}$ and only have polynomial quantum gates. Moreover, classical truncated differential and boomerang attacks are unable to consider the influence of key scheduling in the attack model of single-key, but the proposed algorithms incorporate the key scheduling into the operator $U_{Enc^t}$ and thus fully consider the impact of the key scheduling.

For further research, reducing the quantum complexity of the proposed algorithms is a meaningful direction. It is also a interesting direction to explore the possible applications of more quantum algorithms in other cryptanalytic tools, like integral and algebraic attacks. We believe 
the study of quantum cryptanalysis is crucial for designing quantum-secure cryptosystems to prepare for the arrival of quantum computers.

\section*{Acknowledgement}
This research was supported by Beijing Natural Science Foundation (no.4234084) and the Open Research Fund of Key Laboratory of Cryptography of Zhejiang Province (no. ZCL21012).

\end{document}